\documentclass{llncs}

\usepackage{amsmath,amsfonts,amssymb}
\usepackage{pst-all}
\usepackage{epsfig}
\usepackage{epsf}
\usepackage{here}
\usepackage{graphicx}
\usepackage{graphics}
\usepackage{subfigure}
\usepackage{slashbox}
\usepackage{multirow}
\usepackage{shadow}
\usepackage{psboxit}

\newcommand{\PTS}[2]{\ensuremath{{\Gamma}_{#1,#2}}}

\newcommand{\R}{\ensuremath{\mathbb{R}}}
\newcommand{\LF}{\ensuremath{\lambda}}

\def\qed{\relax\ifmmode\hskip2em \Box\else\unskip\nobreak\hskip1em $\Box$\fi}

\newcommand{\refFigure}[1]{Fig.~\ref{#1}}
\newcommand{\refProposition}[1]{Proposition~\ref{#1}}
\newcommand{\refTheorem}[1]{Theorem~\ref{#1}}
\newcommand{\Equ}[1]{Eq.~(\ref{#1})}

\newcommand{\Conv}[1]{\ensuremath{\mathrm{conv}(#1)}}

\begin{document}

\pagestyle{empty}

\mainmatter

\title{Convex shapes and convergence speed of discrete tangent estimators}

\author{Jacques-Olivier Lachaud  \and Fran\c{c}ois de Vieilleville}

\authorrunning{Lachaud et al.}

\tocauthor{Jacques-Olivier Lachaud, Fran\c{c}ois de Vieilleville (LaBRI, Universit\'{e} Bordeaux 1)}

\institute{LaBRI, Univ. Bordeaux 1, 351 cours de la Lib\'{e}ration, 33405 Talence Cedex, France }



\maketitle

\newcommand{\BigO}[1]{\ensuremath{ \mathcal{O}( #1 )}}
\newcommand{\DG}[0]{\ensuremath{\mathrm{Dig}}}
\newcommand{\Dig}[1]{\ensuremath{\DG_{#1}}}
\newcommand{\GDig}[1]{\ensuremath{\mathrm{D}^G_{#1}}}
\newcommand{\SF}[0]{\ensuremath{\mathbb{F}}}
\newcommand{\GQ}[0]{\ensuremath{G}} 
\newcommand{\DE}[1]{\ensuremath{\mathcal{#1}}} 
\newcommand{\Cite}[1]{\cite{#1}}
\newcommand{\CSF}[0]{\ensuremath{\mathbb{F}_{c}^3}}
\newcommand{\CSFM}[0]{\ensuremath{\mathbb{F}_{c}^3(M)}}
\newcommand{\Z}[0]{\ensuremath{\mathbb{Z}}}

\newcommand{\refClaim}[1]{Claim~\ref{#1}}
\spnewtheorem{myclaim}[theorem]{Claim}{\bfseries}{\itshape}

\newcommand{\TANEDGE}[1]{\ensuremath{\mathcal{T}^\mathrm{edge}_M}}
\newcommand{\TANMS}[1]{\ensuremath{\mathcal{T}^\mathrm{ms}_M}}

\begin{abstract}
  Discrete geometric estimators aim at estimating geometric
  characteristics of a shape with only its digitization as input data.
  Such an estimator is multigrid convergent when its estimates tend
  toward the geometric characteristics of the shape as the
  digitization step $h$ tends toward 0.  This paper studies the
  multigrid convergence of tangent estimators based on maximal digital
  straight segment recognition. We show that such estimators are
  multigrid convergent for some family of convex shapes and that their
  speed of convergence is on average
  $\BigO{h^{\frac{2}{3}}}$. Experiments confirm this result and
  suggest that the bound is tight.
\end{abstract}

\section{Introduction}

The problem of estimating geometric quantities of digitized shapes
like area, perimeter, tangents or curvatures is an important and
active research field. Discrete geometric estimators have indeed many
applications in vision, shape analysis and pattern recognition. It is
however difficult to compare objectively their respective accuracy,
since for a given shape there exists infinitely many shapes with the
same digitization. In this paper, we are mainly interested by the {\em
  multigrid convergence} property of some estimators, which is one of
the few existing objective criteria. This property ensures that a
better resolution brings a better approximation. For concrete
applications, the {\em speed of convergence} is a very important
criterion too, since it has a huge impact on their accuracy at
standard resolutions.

Formally, taking the definitions of \cite{Klette04}, let $\Dig{h}$ be
some digitization process of step $h$. Let $\SF$ be a family of shapes
in $\R^2$ and let $\GQ$ be a geometric quantity defined for all $X \in
\SF$. A {\em discrete estimator of $\GQ$} is a map that associates to
a digitization $\Dig{h}(X)$ an estimation of $\GQ(X)$. A discrete
estimator \DE{\GQ} is {\em multigrid convergent} toward ${\GQ}$ for
$\SF$ and $\DG$ iff, for any $X \in \SF$, there exists some $h_X>0$
for which
\begin{equation*} \textstyle{
  \forall 0<h<h_X, | \DE{\GQ}(\Dig{h}(X)) - \GQ(X) | \le \tau(h) ~, }
\end{equation*}
where $\tau :\R^+ \rightarrow \R^{+*}$ has limit value 0 at
$h=0$. This function defines the {\em speed of convergence} of
$\DE{\GQ}$ toward $\GQ$.

For instance, for the family of plane convex bodies with
$C^3$-boundary and positive curvatures, denoted later on by $\CSF$,
the best known speed of convergence for an area and or a moment
estimator is $\BigO{h^\frac{15}{11}}$ \cite{Huxley90,Klette00}, for a
perimeter estimator it is $\BigO{h}$ \cite{Sloboda98}. There are fewer
results concerning local geometric quantities like tangent direction
or curvature. The first works on this topic were presented in
\cite{Coeurjolly02a}, where some evidence of the convergence of tangent
estimators based on digital straight segment (DSS) recognition were
given. We report the recent work of \Cite{IVC06} which establishes the
multigrid convergence of tangent direction estimators based on maximal
DSS recognition, with an average speed of convergence of
$\BigO{h^{\frac{1}{3}}}$. The recent result of
\cite{deVieilleville05a} has confirmed that there is yet no curvature
estimator proven to be multigrid convergent.

In this paper, we prove a new upper bound for the average speed of
convergence of discrete tangent estimators which are based on maximal
DSS recognition around the point of interest
\cite{Feschet99,Lachaud05a}. This new bound of
$\BigO{h^{\frac{2}{3}}}$ was suggested by the experimental study of
\cite{Lachaud05a} and enhances the previous bound of
$\BigO{h^{\frac{1}{3}}}$ \Cite{IVC06}. The proof of this enhanced
result, obtained for shapes in $\CSF$, follows these steps:
\begin{description}
\item[Section~2] Digitizations of convex shapes are digital convex
  polygons (CDP). We achieve thus a better localization of the shape
  boundary with respect to the digitized boundary
  (\refProposition{prop:poserror}).
\item[Section~3] We first claim that the DSS characteristics have an
  average asymptotic behaviour depending on their length
  (\refClaim{claim:asymptotic-pqdelta}). The asymptotic edge length of
  CDP is recalled (\refTheorem{thm:Balog:nbV1}), which induces a
  superlinear localization of the shape boundary in
  $\BigO{h^{\frac{4}{3}}}$ (\refProposition{prop:poserrorcv}). With
  these relations, the slope of digital edges of CDP is shown to be
  multigrid convergent to the tangent direction of a nearby boundary
  point with an average speed of $\BigO{h^{\frac{2}{3}}}$
  (\refProposition{prop:edgeconvergence}).
\item[Section~4] The behaviour of maximal DSS is identical to the
  behavior of digital edges, since any maximal DSS contains at least
  one digital edge (\refProposition{prop:ms-includes-edge}). The
  average speed of convergence of tangent estimators based on
  maximal DSS is thus achieved (\refTheorem{thm:maxDssConvergence}).
\end{description}
We will then conclude and open some perspectives in Section~5.

\section{Preliminary definitions and first properties}

\subsection{Digitization and convex digital polygon}

Let $S$ be some subset of $\R^2$. Its {\em Gauss
digitization} of grid step $h$ is defined as $\GDig{h}(S) = S \cap
h\Z \times h\Z $. Thus, the considered digitized
objects are subsets of the rescaled digital plane $h\Z \times h\Z$.

A \emph{convex digital polygon (CDP)} $\Gamma$ is a subset of the
digital plane $h\Z \times h\Z$ with a single 4-connected component
equal to the Gauss digitization of its convex hull, i.e. $\Gamma =
\GDig{h}(\Conv{\Gamma})$. Its {\em vertices} $(V_i)_{i=1..e}$ form the
minimal subset for which $\Gamma=\GDig{h}(\Conv{V_1, \ldots,V_e})$
(vertices are ordered clockwise). The points of $\Gamma$ which are
8-adjacent to some point not in $\Gamma$ form the {\em border} of
$\Gamma$. It is a 4-connected digital path that visits every $V_i$ in
order. When moving clockwise, any subpath $\PTS{V_i}{V_{i+1}}$ is
called a {\em digital edge} of $\Gamma$, while the Euclidean straight
segment $V_i V_{i+1}$ is an {\em edge} of $\Gamma$.

For small enough grid steps, Gauss digitizations of
any shape in $\CSF$ are convex digital polygons. This result will be
admitted throughout the paper.

\subsection{Standard line, digital straight segment, maximal segments}
\begin{definition} \cite{Rev91} The set of points $(x,y)$ of the
  digital plane $\Z^2$ verifying $\mu \le ax-by < \mu + |a|+|b|$, with
  $a$, $b$ and $\mu$ integer numbers, is called the {\em standard
    line} with slope $a/b$ and shift $\mu$. Standard lines are the
  4-connected discrete lines.
\end{definition}

The quantity $ax-by$ is called the {\em remainder} of the line. Points
whose remainder is $\mu$ (resp. $\mu+|a|+|b|-1$) are called upper
(resp. lower) leaning points. Any finite connected portion of a
standard line is called a {\em digital straight segment (DSS)}. Its
{\em characteristics} are the slope $a/b$ and the shift $\mu$ of the
standard line containing it with smallest $|a|+|b|$.

Most of the results demonstrated here are directly transferable for
8-connected curves since there is a natural bijective transformation
between standard and naive digital lines.  In the paper, all the
reasoning is made in the first octant, but it extends naturally to the
whole digital plane.

\subsection{Recursive decomposition of DSS}
We here recall a few properties about \emph{patterns} composing DSS
and their close relations with continued fractions. They constitute a
powerful tool to describe discrete lines with rational slopes
\cite{Berstel97,HarWri60}.  Without loss of generality all definitions
and propositions stated below hold for standard lines and DSS with
slopes in the first octant (e.g. $\frac{a}{b}$ with $ 0 \leq a \leq b
$). In the first octant, only two Freeman moves are possible, 0 is a
step to the right and 1 is a step up, 4-connected digital paths can be
expressed as words of $\{0,1\}^{*}$.
\begin{definition}
  Given a standard line $(a,b,\mu)$, we call \emph{pattern} of
  characteristics $(a,b)$ the word that is the succession of Freeman
  moves between any two consecutive upper leaning points. The Freeman
  moves defined between any two consecutive lower leaning points is
  the previous word read from back to front and is called the
  \emph{reversed pattern}.
\end{definition}
Since a DSS has at least either two upper or two lower leaning points,
a DSS $(a,b,\mu)$ contains at least one \emph{pattern} or one
\emph{reversed pattern} of characteristics $(a,b)$. There exists
recursive transformations for computing the \emph{pattern} of a
standard line from the \emph{simple continued fraction} of its slope
(see \cite{Berstel97}, \cite{Voss93} Chap.~4 and \cite{Klette04}
Chap.~9), here Berstel approach better suits our purpose.

A rational slope $z$ in $]0,1]$ can be written uniquely as the
continued fraction
\[
\textstyle{
z = 0 + \cfrac{1}{ u_1 + \cfrac{1}{\ldots +
     \cfrac{1}{u_n}}} }
\]
and is conveniently denoted $[0,u_{1}, \ldots, u_n]$.  The
$u_{i}$ are called the \emph{partial coefficients} and the continued
fraction formed with the $k$ first \emph{partial coefficient} is
said to be a \emph{$k$-th convergent} of $z$ and is a rational numbers
denoted by $z_{k}$.  The \emph{depth} of a \emph{$k$-th convergent}
equals $k$. We conveniently denote $p_{k}$ the numerator (resp.
$q_{k}$ the denominator) of a \emph{$k$-th convergent}.

We recall a few more relations regarding the way convergents are related and
which will be used later on in this paper:
\begin{eqnarray} 
 & \textstyle{ \forall k \geq 1 } & \textstyle{ \quad p_{k}q_{k-1} - p_{k-1}q_{k} = (-1)^{k+1} \label{pattern:rec:dif} }\\
\textstyle{ p_{0}=0 \quad p_{-1}=1 \quad \quad \quad } & \textstyle{ \forall k \geq 1 } &  \textstyle{ \quad p_{k}=u_{k}p_{k-1}+p_{k-2} \label{pattern:rec:num} }\\
\textstyle{ q_{0}=1 \quad q_{-1}=0 \quad \quad \quad } & \textstyle{ \forall k \geq 1 } &  \textstyle{ \quad q_{k}=u_{k}q_{k-1}+q_{k-2} \label{pattern:rec:den} }
\end{eqnarray}
Given a rational slope between $0$ and $1$ its continued fraction is
finite and for each $i$, $u_{i}$ is a strictly positive integer. In
order to have a unique writing we consider that the last \emph{partial
  coefficient} is greater or equal to two; except for slope $1 =
[0,1]$.

Let us now explain how to compute the \emph{pattern} associated with a
rational slope $z$ in the first octant. Let us define $E$ a mapping
from the set of positive rational number smaller than one onto the
Freeman-move's words. More precisely: $E(z_0) = 0$, $E(z_1) =
0^{u_1}1$ and the other values are expressed recursively:
\begin{eqnarray} 
  \textstyle{  E(z_{2i+1})} & \textstyle{=} &   \textstyle{E(z_{2i})^{u_{2i+1}} E(z_{2i-1})}
  \label{pattern:rec:odd}\\
  \textstyle{ E(z_{2i})} & \textstyle{ =} &  \textstyle{ E(z_{2i - 2}) E(z_{2i-1})^{u_{2i}} }
  \label{pattern:rec:even} 
\end{eqnarray}

It has been shown that this mapping constructs the pattern $(a,b)$ for
any rational slope $z=\frac{a}{b}$. Fig. \ref{fig:aop-pattern_I}
illustrates the construction of an odd pattern using the mapping $E$.
The Minkowski $\mathcal{L}^{1}$ length of $E(z_{k})$ equals $p_{k}+
q_{k}$ and can be computed recursively using \Equ{pattern:rec:num} and
(\ref{pattern:rec:den}). Moreover we recall that any digital edge is a
pattern or a succession of the same pattern, its {\em digital
parameters} are its slope, denoted by $\frac{p}{q}$, and the number
$\delta$ of repetitions of the pattern $E(\frac{p}{q})$.

\begin{figure}[tbp]
  \begin{center}
    
    \begin{picture}(0,0)%
\includegraphics{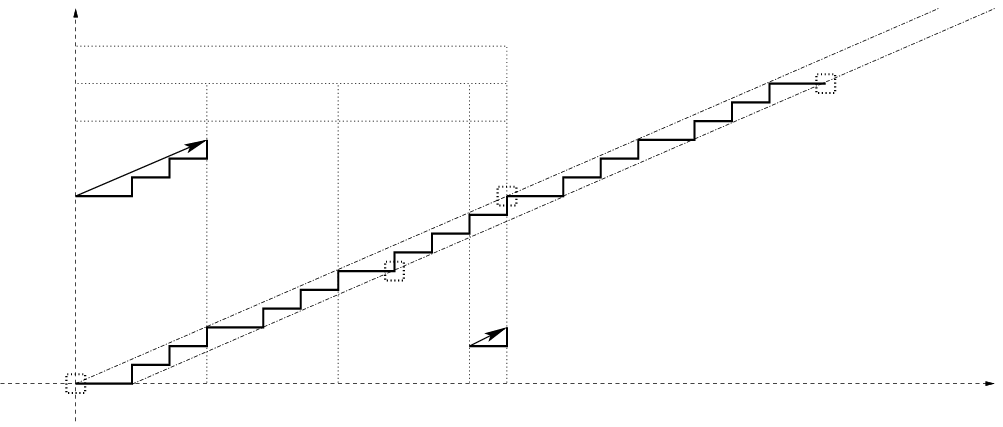}%
\end{picture}%
\setlength{\unitlength}{1184sp}%
\begingroup\makeatletter\ifx\SetFigFont\undefined%
\gdef\SetFigFont#1#2#3#4#5{%
  \reset@font\fontsize{#1}{#2pt}%
  \fontfamily{#3}\fontseries{#4}\fontshape{#5}%
  \selectfont}%
\fi\endgroup%
\begin{picture}(15924,6706)(1789,-7055)
\put(2101,-961){\makebox(0,0)[lb]{\smash{{\SetFigFont{7}{8.4}{\familydefault}{\mddefault}{\updefault}{\color[rgb]{0,0,0}$Y$}%
}}}}
\put(16876,-6961){\makebox(0,0)[lb]{\smash{{\SetFigFont{7}{8.4}{\familydefault}{\mddefault}{\updefault}{\color[rgb]{0,0,0}$X$}%
}}}}
\put(7651,-5161){\makebox(0,0)[lb]{\smash{{\SetFigFont{7}{8.4}{\familydefault}{\mddefault}{\updefault}{\color[rgb]{0,0,0}$L_{1}$}%
}}}}
\put(14401,-2161){\makebox(0,0)[lb]{\smash{{\SetFigFont{7}{8.4}{\familydefault}{\mddefault}{\updefault}{\color[rgb]{0,0,0}$L_{2}$}%
}}}}
\put(10201,-2761){\makebox(0,0)[lb]{\smash{{\SetFigFont{7}{8.4}{\familydefault}{\mddefault}{\updefault}{\color[rgb]{0,0,0}$U_{2}$}%
}}}}
\put(3151,-2011){\makebox(0,0)[lb]{\smash{{\SetFigFont{7}{8.4}{\familydefault}{\mddefault}{\updefault}{\color[rgb]{0,0,0}$E(z_{2})$}%
}}}}
\put(5251,-2011){\makebox(0,0)[lb]{\smash{{\SetFigFont{7}{8.4}{\familydefault}{\mddefault}{\updefault}{\color[rgb]{0,0,0}$E(z_{2})$}%
}}}}
\put(7351,-2011){\makebox(0,0)[lb]{\smash{{\SetFigFont{7}{8.4}{\familydefault}{\mddefault}{\updefault}{\color[rgb]{0,0,0}$E(z_{2})$}%
}}}}
\put(9451,-2011){\makebox(0,0)[lb]{\smash{{\SetFigFont{7}{8.4}{\familydefault}{\mddefault}{\updefault}{\color[rgb]{0,0,0}$E(z_{1})$}%
}}}}
\put(5251,-2761){\makebox(0,0)[lb]{\smash{{\SetFigFont{7}{8.4}{\familydefault}{\mddefault}{\updefault}{\color[rgb]{0,0,0}$p_{2}$}%
}}}}
\put(10051,-5761){\makebox(0,0)[lb]{\smash{{\SetFigFont{7}{8.4}{\familydefault}{\mddefault}{\updefault}{\color[rgb]{0,0,0}$p_{1}$}%
}}}}
\put(3151,-3811){\makebox(0,0)[lb]{\smash{{\SetFigFont{7}{8.4}{\familydefault}{\mddefault}{\updefault}{\color[rgb]{0,0,0}$q_{2}$}%
}}}}
\put(1951,-6061){\makebox(0,0)[lb]{\smash{{\SetFigFont{7}{8.4}{\familydefault}{\mddefault}{\updefault}{\color[rgb]{0,0,0}$U_{1}$}%
}}}}
\put(2101,-6961){\makebox(0,0)[lb]{\smash{{\SetFigFont{7}{8.4}{\familydefault}{\mddefault}{\updefault}{\color[rgb]{0,0,0}$O$}%
}}}}
\put(9451,-6211){\makebox(0,0)[lb]{\smash{{\SetFigFont{7}{8.4}{\familydefault}{\mddefault}{\updefault}{\color[rgb]{0,0,0}$q_{1}$}%
}}}}
\put(3901,-1411){\makebox(0,0)[lb]{\smash{{\SetFigFont{7}{8.4}{\familydefault}{\mddefault}{\updefault}{\color[rgb]{0,0,0}$E(z_{3})=[0,2,3,3]=\frac{10}{23}$}%
}}}}
\put(13051,-4261){\makebox(0,0)[lb]{\smash{{\SetFigFont{7}{8.4}{\familydefault}{\mddefault}{\updefault}{\color[rgb]{0,0,0}$z_{3} = [0,2,3,3]$}%
}}}}
\end{picture}%

    \caption{A digital straight segment of characteristics $(10,23,0)$
      with an odd depth slope, taken between the origin and its second
      lower leaning point }
    \label{fig:aop-pattern_I}
  \end{center}
\end{figure}

\subsection{Localization accuracy of digitized convex shapes}

One can expect that the boundary of a convex shape is approximately at
distance $h$ from the border of its digitization of grid step $h$. In
fact, for a convex shape $S$, its {\em boundary $\partial S$} is much
closer to the convex hull of $\GDig{h}(S)$ than $h$. A better
localization of the shape is thus possible, as stated below and
illustrated in \refFigure{fig:ULUConstrains}:

\begin{proposition} \label{prop:poserror} 
  Let $S$ be a convex shape such that $\GDig{h}(S)$ is a CDP $\Gamma$
  for some $h$. Consider an edge $V_i V_{i+1}$ of $\Gamma$ with slope
  in the first octant. Then any point of the boundary $\partial S$
  above the straight segment $V_i V_{i+1}$ has a vertical distance to
  it no greater than $\frac{h}{q_{n-1}}$, where the slope of this edge
  is the irreducible fraction $\frac{p_{n}}{q_{n}}$. 

\end{proposition}

The proof of this proposition can be found in \cite{Lachaud06a}.

  \begin{figure}[tbp]
    \begin{center}

      \begin{picture}(0,0)%
\includegraphics{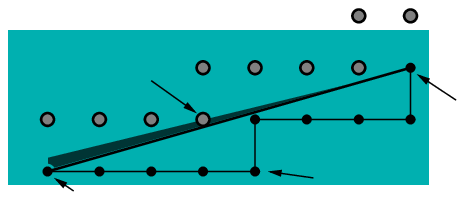}%
\end{picture}%
\setlength{\unitlength}{2960sp}%
\begingroup\makeatletter\ifx\SetFigFont\undefined%
\gdef\SetFigFont#1#2#3#4#5{%
  \reset@font\fontsize{#1}{#2pt}%
  \fontfamily{#3}\fontseries{#4}\fontshape{#5}%
  \selectfont}%
\fi\endgroup%
\begin{picture}(3845,1401)(2440,-10939)
\put(4600,-10755){\makebox(0,0)[lb]{\smash{{\SetFigFont{10}{12.0}{\familydefault}{\mddefault}{\updefault}{\color[rgb]{0,0,0}$L$}%
}}}}
\put(3397,-9967){\makebox(0,0)[lb]{\smash{{\SetFigFont{10}{12.0}{\familydefault}{\mddefault}{\updefault}{\color[rgb]{0,0,0}$P$}%
}}}}
\put(3106,-10880){\makebox(0,0)[lb]{\smash{{\SetFigFont{10}{12.0}{\familydefault}{\mddefault}{\updefault}{\color[rgb]{0,0,0}$V_{1}$}%
}}}}
\put(5554,-10257){\makebox(0,0)[lb]{\smash{{\SetFigFont{10}{12.0}{\familydefault}{\mddefault}{\updefault}{\color[rgb]{0,0,0}$V_{2}$}%
}}}}
\put(5333,-9826){\makebox(0,0)[lb]{\smash{{\SetFigFont{10}{12.0}{\familydefault}{\mddefault}{\updefault}{\color[rgb]{0,0,0}$P'$}%
}}}}
\put(2440,-10523){\makebox(0,0)[lb]{\smash{{\SetFigFont{10}{12.0}{\familydefault}{\mddefault}{\updefault}{\color[rgb]{0,0,0}$P''$}%
}}}}
\end{picture}%

    \end{center} \vspace{-0.3cm}

    \caption{Two local constraints for the real underlying convex shape  
      \label{fig:ULUConstrains}}
  \end{figure}

\section{Asymptotic behaviour of edges of digitized shapes}

We study here asymptotic properties of edges of digitized convex
shapes. Their average length is first exhibited and an experimental
study of the digital parameters of edges is presented. The direction of
digital edges is shown to converge toward the tangent direction. 

\subsection{Asymptotic number and digital parameters of edges}

Let $S$ be some shape of $\CSF$. We have the following theorem:

\begin{theorem}{(Adapted from \cite{Balog1991}, Theorem~2)}
  \label{thm:Balog:nbV1} 
  For a small enough $h$, the Gauss digitization $\GDig{h}(S)$ of $S$
  is a CDP and its number of edges $n_e(\GDig{h}(S))$ 
  satisfies:
  \[ \textstyle{ c_{1}(S) \frac{1}{h^{\frac{2}{3}}} \leq
    n_{e}(\GDig{h}(S)) \leq c_{2}(S) \frac{1}{h^{\frac{2}{3}}} } \]
  where the constants $c_{1}(S)$ and $c_{2}(S)$ depend on extremal
  bounds of the curvatures along $S$. Hence for a disc $c_{1}$ and
  $c_{2}$ are absolute constants.
\end{theorem}

As an immediate corollary, the average Minkowski $\mathcal{L}^1$
length of edges grows as $\Theta(h^{\frac{2}{3}})$. The question is:
what is the average behavior of the digital parameters $p$, $q$ and
$\delta$ of edges, knowing that the average digital length
$l=\delta(p+q)$ satisfies
$\Theta(h^{\frac{2}{3}}/h)=\Theta(1/h^{\frac{1}{3}})$ ? Since the slope
of edges should tend toward the slope of points on the boundary of $S$
and since almost all these points have irrational slope, $p$ and $q$
should tend toward infinity almost everywhere with a bounded
$\delta$. This observation is confirmed by experiments, as illustrated
on \refFigure{fig:pqdCircle}, which plots the means and standard
deviations of $\frac{q}{l}$ and $\delta$ for edges on finer and finer
digitizations of a disk. It is clear that $q$ (but also $p$) satisfies
the same asymptotic law as $l$ while $\delta$ remains bounded on
average. We make hence the following claim:

\begin{figure}[tbp]
  \begin{center}
    \epsfig{file=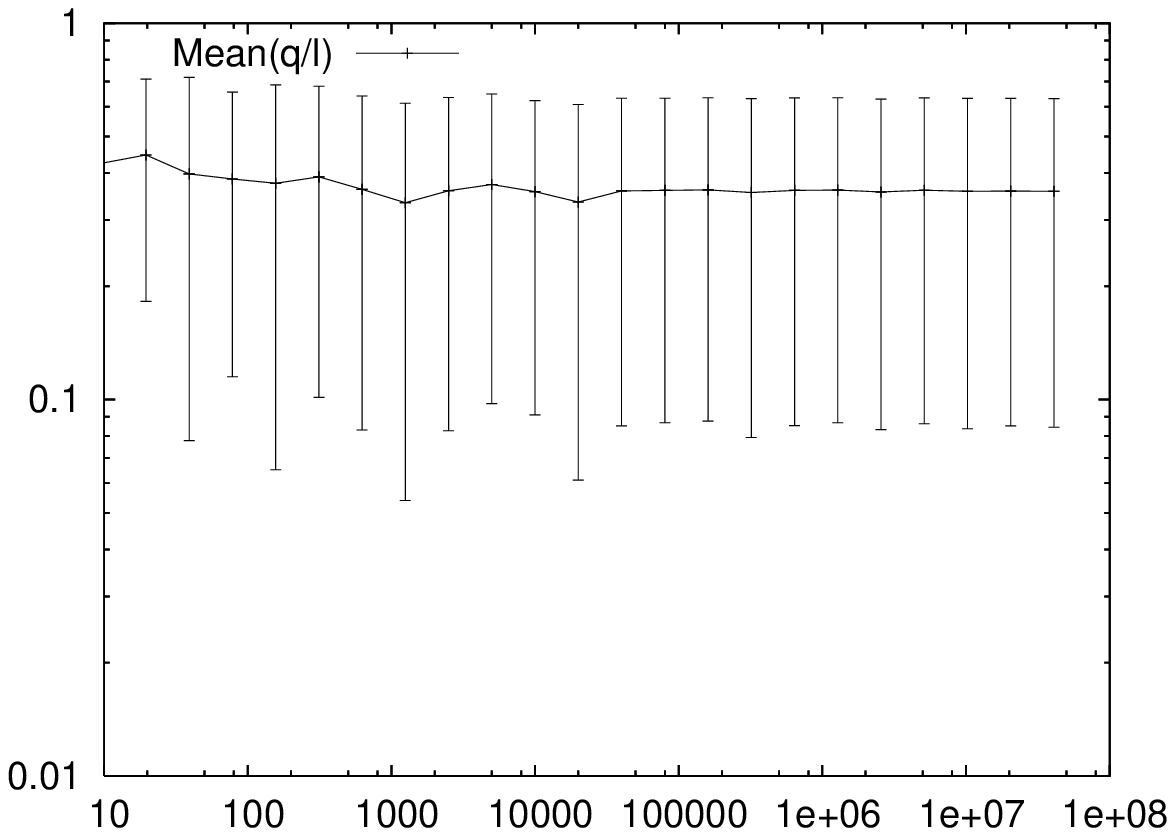,width=0.49\textwidth}
    \epsfig{file=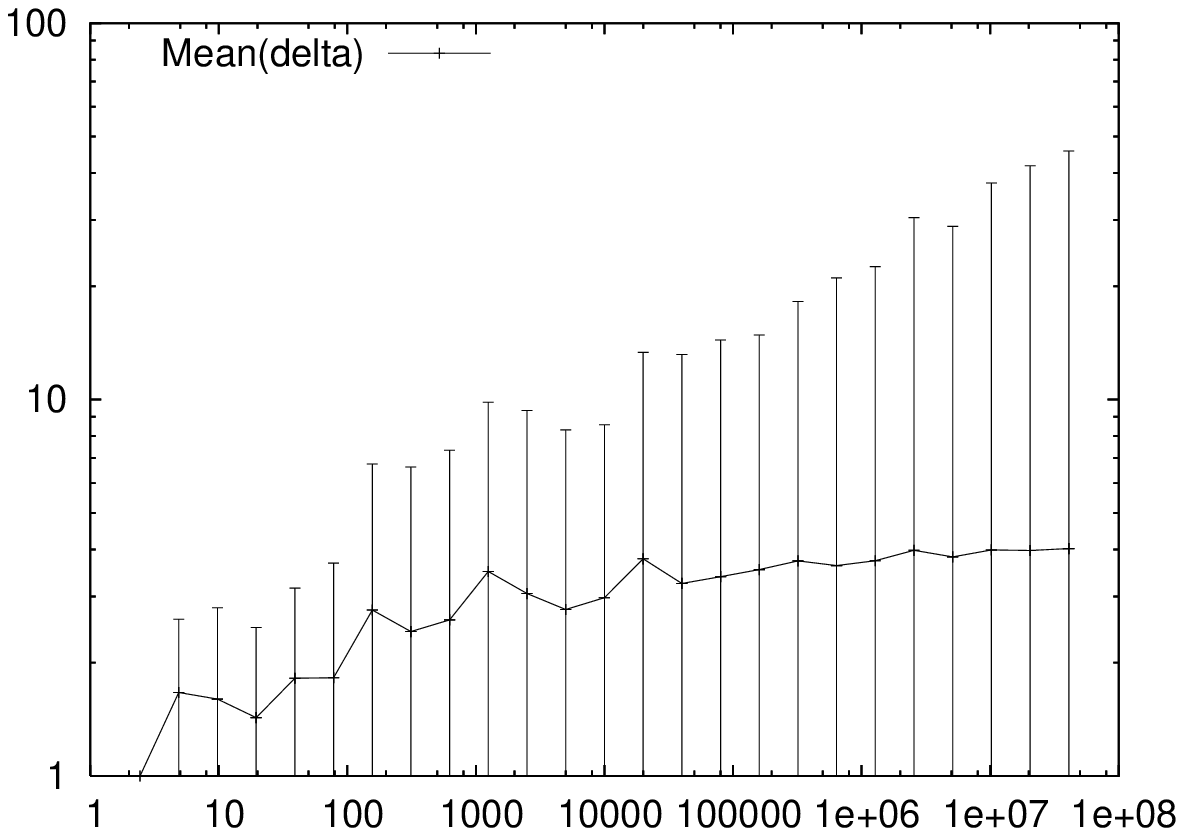,width=0.49\textwidth}
  \end{center}
  \caption{Plots in log-space of the means of digital parameters
    $\frac{q}{l}$ (left) and $\delta$ (right) for the edges of Gauss
    digitizations of a disk, as functions of the inverse $\frac{1}{h}$
    of the grid step. Standard deviations are symbolized with
    errorbars
    \label{fig:pqdCircle} }
\end{figure} 

\begin{myclaim} \label{claim:asymptotic-pqdelta}
  On average, the digital parameters $p$, $q$ of the edges of
  $\GDig{h}(S)$ with $S \in \CSF$ grow as $\Theta(\frac{1}{h^{\frac{1}{3}}})$,
  whereas $\delta$ is bounded, when $h$ tends toward 0.
\end{myclaim}

This claim induces a first result about the convergence speed of the
localization constraints of \refProposition{prop:poserror}
($q_{n-1}$ has the same asymptotic law as $q_n$).

\begin{proposition} \label{prop:poserrorcv} 
  Let $S \in \CSF$ and $\GDig{h}(S)$ its Gauss digitization. For a
  small enough $h$, in each octant, the vertical distance between any
  edge of $\GDig{h}(S)$ and $\partial S$ is bounded. On average, this
  bound is $\BigO{h^{\frac{4}{3}}}$ (and never worse than
  $\BigO{h}$).
\end{proposition}

\subsection{Convergence of tangent estimation based on edge direction}

Let $M$ be a point on the boundary of a shape $S \in \CSF$ and let
$\theta_M(S)$ be the tangent direction of $\partial S$ at $M$. We
propose to estimate $\theta_M(S)$ as the arctangent of the slope of
the digital edge of the CDP $\GDig{h}(S)$ lying below $M$ (for the
first octant; definitions for other octants are straightforward). We
denote this discrete estimator by $\TANEDGE{M}$. Assuming
\refClaim{claim:asymptotic-pqdelta}, we have the multigrid convergence
of this tangent estimator:

\begin{proposition} \label{prop:edgeconvergence} 
  Let $M$ be a point in the plane and let $\CSFM$ be the subset of
  shapes $S$ of $\CSF$ with $M \in \partial S$. The discrete estimator
  $\TANEDGE{M}$ at point $M$ is multigrid convergent toward the
  tangent direction $\theta_M$ for $\CSFM$ and Gauss
  digitization. Furthermore, its speed of convergence is on average
  \BigO{h^{\frac{2}{3}}}.
\end{proposition}
\begin{proof}
  We have to prove that for a shape $S \in \CSFM$, there exists some
  grid step $h_S$ for which $\forall 0< h <h_S,
  |\TANEDGE{M}(\GDig{h}{S})-\theta_M(S)| \le \tau(h)$. Without loss of
  generality, we assume that the tangent direction $\theta_M(S)$ is in
  the first octant and we locally parameterize $\partial S$ as
  $(x,f(x))$, setting $M$ as the origin of the coordinate axes
  ($(x_M,y_M)=(0,0)$). Let $h_0$ be the digitization step below which
  $\GDig{h}(S)$ is a CDP.

  Since the slope of $\partial S$ around $M$ is in the first octant,
  there exists some $h_S < h_0$ such that for any $0<h<h_S$, the
  vertical straight segment going down from $M$ intersect some edge
  $V_i V_{i+1}$ of $\GDig{h}(S)$. Let $\delta$, $p$ and $q$ be the
  digital parameters of this edge, let $z$ be its slope $\frac{p}{q}$
  and let $n$ be the depth of $z$ (i.e. $z =
  \frac{p_{n}}{q_{n}}$). The horizontal distance between $M$ and
  either $V_i$ or $V_{i+1}$ is necessarily greater than half the
  horizontal length of $V_i V_{i+1}$. We assume $V_{i+1}$ maximizes this
  distance, without loss of generality. We get:
  \begin{equation} 
    \label{eq:bound-on-x_right} \textstyle{
    \frac{ h \delta q_{n}}{2}\leq x_{v_{i+1}} \leq h \delta q_{n} }
  \end{equation}

  From \refProposition{prop:poserror} we have for any $x \in
  [x_{v_i},x_{v_{i+1}}]$:
  \begin{equation}
    \label{eq:bound-on-f} \textstyle{
    z x - \frac{h}{q_{n-1 }} \leq f(x) \leq z x + \frac{h}{q_{n- 1}} ~.}
  \end{equation}
  Inserting Taylor expansion of $f(x)$ about $x=0$ in
  \Equ{eq:bound-on-f} induces:
  \begin{eqnarray}
    \label{eq:bound-on-f-Taylor}
    \textstyle{ | z - f'(0) |  \leq \frac{h}{x q_{n - 1 } } + \BigO{x} ~.}
  \end{eqnarray}
  Setting $x=x_{v_{i+1}}$ in \Equ{eq:bound-on-f-Taylor} and using
  both sides of \Equ{eq:bound-on-x_right} gives the inequality
  \begin{equation}
    \label{eq:bound-on-f-2} \textstyle{
    | z - f'(0) | \leq \frac{2 h}{h \delta q_{n} q_{n - 1 } } + \BigO{
      h \delta q_{n}} ~.}
  \end{equation}
  We notice that $z=\tan(\TANEDGE{M}(\GDig{h}(S)))$ and that
  $f'(0)=\tan(\theta_M(S))$. Moreover, for any angle $s,t \in
  [0,\frac{\pi}{4}]$, we have $|s-t| \leq |\tan(s)-\tan(t)|$. With
  these two remarks, \Equ{eq:bound-on-f-2} implies:
  \begin{equation} \textstyle{
    \label{eq:multigrid-cvg-1}   
    | \TANEDGE{M}(\GDig{h}(S)) - \theta_M(S) | \leq \frac{2}{\delta q_{n} q_{n - 1 } } + \BigO{h \delta q_{n}} ~. }
  \end{equation}
  Since $\delta q_n$ is no smaller than half the edge length, it
  follows from \refTheorem{thm:Balog:nbV1} that the dominant term
  $\frac{2}{\delta q_{n} q_{n - 1 }}$ is at least some
  $\BigO{h^\frac{1}{3}}$ on average. In \cite{IVC06} it is also shown
  that there is no edge of bounded length as $h$ tends toward
  $0$. Then the right part of \Equ{eq:multigrid-cvg-1} tends toward
  0. The multigrid convergence is thus shown. At last, assuming
  \refClaim{claim:asymptotic-pqdelta}, \Equ{eq:multigrid-cvg-1} then
  induces
  \begin{equation} \textstyle{
    \label{eq:multigrid-cvg-2} 
      | \TANEDGE{M}(\GDig{h}(S)) - \theta_M(S) | \leq \BigO{h^{\frac{2}{3}}} ~,}
  \end{equation}
  which indicates that the average speed of convergence of
  $\TANEDGE{M}$ is $\BigO{h^{\frac{2}{3}}}$.  \qed
\end{proof}

\section{Tangent estimators based on maximal DSS recognition}

This section discusses the convergence speed of discrete tangent
estimators based on maximal digital straight segment (maximal segment)
recognition. Along a digital path, maximal segments are the
inextensible digital straight segments, otherwise said adding the next
point to the front or to the back constructs a set of digital points
that no standard line contains. The set of all maximal segments of a
digital path can be extracted efficiently in time linear with its
number of points \cite{Feschet99}. Maximal segments have deep links
with edges of convex hulls \cite{deVieilleville05a,Doerksen04}. 

Estimating the tangent direction at some point is then achieved by
considering specific DSS \cite{Vialard96} or maximal segments
\cite{Feschet99} containing this point. A recent
experimental evaluation \cite{Lachaud05a} has shown that tangent
estimators based on maximal segments are accurate and preserve
convexity properties of the real shape. 

Let $M$ be a point on the boundary of a shape $S \in \CSF$ and let
$\theta_M(S)$ be the tangent direction of $\partial S$ at $M$. We
propose to estimate $\theta_M(S)$ as the arctangent of the slope of
any maximal segment of the CDP $\GDig{h}(S)$ lying below $M$ (for the
first octant; definitions for other octants are straightforward). We
denote this discrete estimator by $\TANMS{M}$. We shall prove in
\refTheorem{thm:maxDssConvergence} that this estimator is multigrid
convergent with average convergence speed of
$\BigO{h^\frac{2}{3}}$. As a corollary, the Feschet-Tougne tangent
estimator \cite{Feschet99}, which uses the most centered maximal
segment around $M$, and the $\lambda$-MST estimator \cite{Lachaud05a},
which makes a convex combination of the directions of all maximal
segments around $M$, are also multigrid convergent with same average
speed.

Before proving this theorem, note first that the experimental
evaluation of the $\lambda$-MST estimator, whose absolute error is
displayed in \refFigure{fig:MST-circle-moy-and-MS-moy}, left, suggests
that its average convergence speed is indeed likely to be in
$\BigO{h^{\frac{2}{3}}}$. Secondly, the average Minkowski length of
maximal segments behave as the average length of the edges, that is in
$\BigO{h^{\frac{2}{3}}}$ as shown on
\refFigure{fig:MST-circle-moy-and-MS-moy}, right (the Minkowski length
is $h$ times the digital length). Thirdly, every maximal segment
contains a digital edge as stated in the following proposition whose
proof can be found in \cite{Lachaud06a}:

\begin{figure}[tbp]
  \begin{center}
    \rotatebox[origin=br]{-90}{\epsfig{file=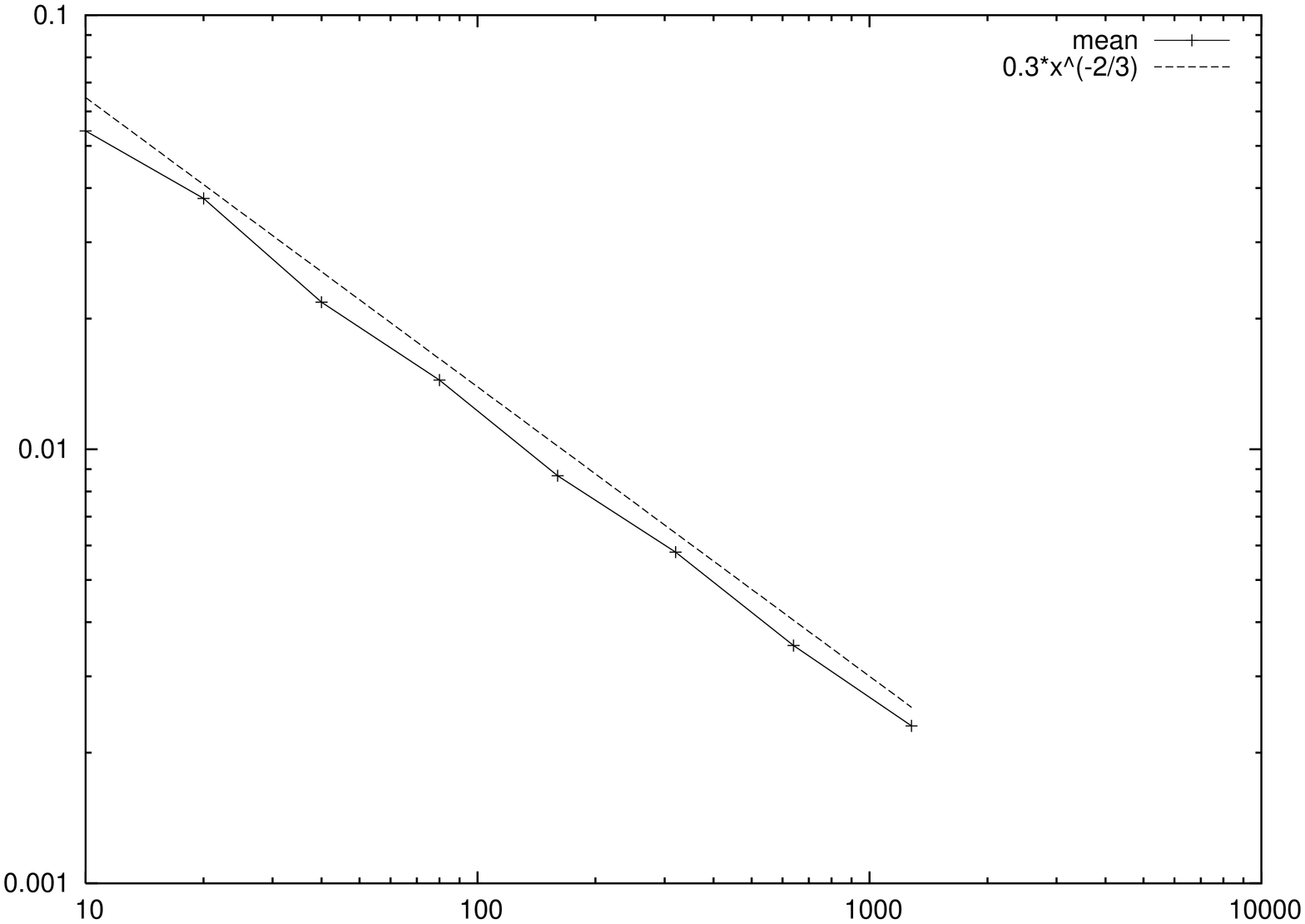,height=0.49\textwidth}}
    \epsfig{file=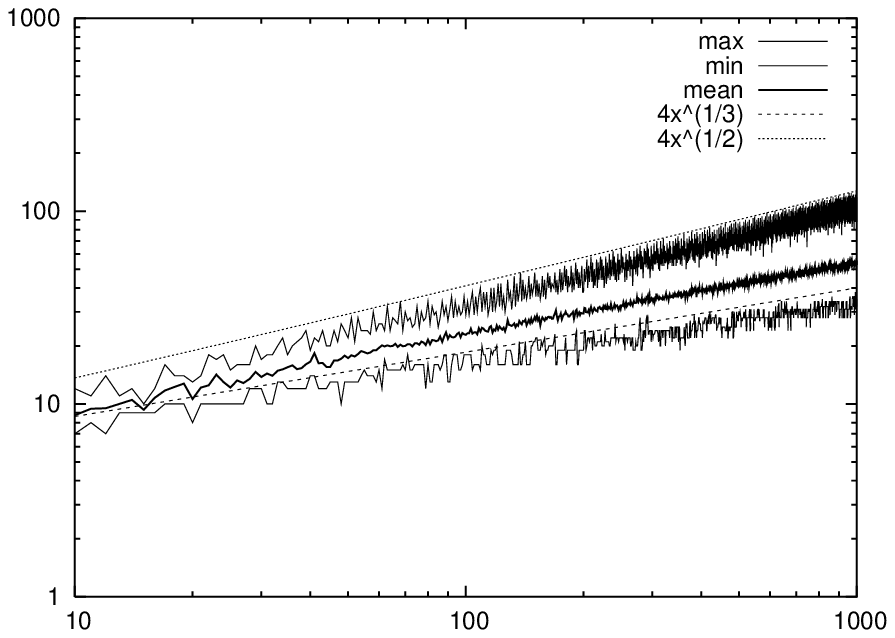,width=0.49\textwidth}
  \end{center}
  \caption{For both plots, the digitized shape is a disk of radius 1
    and the abscissa is the inverse of the digitization step. Left:
    plot in log-space of the mean absolute error between the $\LF$-MST
    tangent direction and the theoretical one (mean taken over all
    digital points).  The convergence speed on this shape is likely to
    be in $\Theta(h^{\frac{2}{3}})$. Right: plot in log-space of the
    digital length of maximal segments, which is on average
    $\Theta(1/h^{\frac{1}{3}})$ }
  \label{fig:MST-circle-moy-and-MS-moy}
\end{figure}

\begin{proposition}  \label{prop:ms-includes-edge}
  On the border of a CDP $\Gamma$, any maximal segment, whose slope is
  some $\frac{p_n}{q_n}$, contains at least one digital edge that has
  either the same slope or its $n-1$-convergent
  $\frac{p_{n-1}}{q_{n-1}}$.
\end{proposition}

We may now state our main result, which relies also on
\refClaim{claim:asymptotic-pqdelta}.
\begin{theorem} \label{thm:maxDssConvergence} 
  Let $M$ be a point in the plane and let $\CSFM$ be the subset of
  shapes $S$ of $\CSF$ with $M \in \partial S$. The discrete estimator
  $\TANMS{M}$ at point $M$ is multigrid convergent toward the tangent
  direction $\theta_M$ for $\CSFM$ and Gauss
  digitization. Furthermore, its speed of convergence is on average
  \BigO{h^{\frac{2}{3}}}.
\end{theorem}

\begin{proof}
  We take the same notations as in the proof of
  \refProposition{prop:edgeconvergence}. We consider a maximal segment
  $MS$ below $M$ with slope $z_n$. According to
  \refProposition{prop:ms-includes-edge}, it contains a digital edge $VV'$
  of slope $z_n$ or $z_{n-1}$, which may not be below $M$. Worst case
  is for a digital edge with slope $z_{n-1}$ and this is the one
  considered below. We have:
  \begin{eqnarray*}
    \textstyle{ |\TANMS{M}(S)-\theta_M(S)| } & \textstyle{ \le } & \textstyle{ |z_n-f'(0)| } \\
    & \textstyle{ \le } & \textstyle{ |z_{n-1}-f'(0)| + |z_n-z_{n-1}| } \\
    & \textstyle{ \le } & \textstyle{ |z_{n-1}-f'(x)| + \BigO{x}+\frac{1}{q_n q_{n-1}} ~, }\\
  \end{eqnarray*}
  using Taylor relation and \Equ{pattern:rec:dif} to get the last
  inequality.

  If we choose some abscissa $x$ such that $(x,f(x))$ is above the
  edge $VV'$, then the proof of \refProposition{prop:edgeconvergence}
  indicates that $|z_{n-1}-f'(x)| \le \BigO{h^{\frac{2}{3}}}$ on
  average. Since $(x,f(x))$ is above $MS$ too, $x$ cannot be greater
  than the length of $MS$, which is also some $\BigO{h^{\frac{2}{3}}}$ on
  average. At last, \refClaim{claim:asymptotic-pqdelta} similarly
  provides $\frac{1}{q_n q_{n-1}}=\Theta(h^{\frac{2}{3}})$. Summing
  all these bounds concludes the proof. \qed
\end{proof}
As an immediate corollary, the $\LF$-MST estimator \cite{Lachaud05a}
and the Feschet-Tougne tangent estimator \cite{Feschet99} have the
same asymptotic behaviour.

\section{Conclusion}

We have studied properties of digital convex polygons and
exhibited several new results (\refProposition{prop:poserror} and
\refProposition{prop:ms-includes-edge}). We have also examined the
asymptotic properties of digital edges on digitized convex shapes
(\refClaim{claim:asymptotic-pqdelta}), which has led to a position
estimator of average convergence speed $\BigO{h^\frac{4}{3}}$
(\refProposition{prop:poserrorcv}) and to a tangent direction
estimator of average convergence speed $\BigO{h^\frac{2}{3}}$
(\refProposition{prop:edgeconvergence}). At last we have shown the new
bound of $\BigO{h^\frac{2}{3}}$ for the average convergence speed of
tangent estimators based on maximal segments
(\refTheorem{thm:maxDssConvergence}), which matches experimental
evaluation.

These results indicate that curvature estimators relying on
surrounding DSS have most probably an error of $\BigO{h^0}$: position
$\equiv f(x)$ in \BigO{h^\frac{4}{3}}, tangent direction $\equiv
f'(x)$ in \BigO{h^\frac{2}{3}}, curvature $\equiv f''(x)$ probably in
$\BigO{1}$ (uncertainty on tangent further divided by
\BigO{h^\frac{2}{3}}). The problem of exhibiting a multigrid
convergent curvature estimator is thus still open. 

Another straightforward extension of this work would be to investigate
the properties of discrete surfaces and estimators based on digital
plane recognition. However since the problem of finding an enclosing
polyhedron with a minimal number of 2-facets has been proven to be
NP-hard (see \cite{Brimkov06}), the problem would get much harder than the 
two dimensional case.


\bibliographystyle{plain}

\bibliography{iscv} 

\end{document}